\documentclass[10pt,conference]{IEEEtran}
\IEEEoverridecommandlockouts

\usepackage{amssymb}
\usepackage[cmex10]{amsmath}
\usepackage{stfloats}
\usepackage{graphicx}
\usepackage{subfigure}
\usepackage{tabularx}
\usepackage{epsfig,epsf,color,balance,cite}
\usepackage{verbatim}
\usepackage{url}
\usepackage{bm}
\usepackage{booktabs}
\usepackage[ruled,linesnumbered]{algorithm2e}
\usepackage{algorithmic}

\newtheorem{lemma}{\bf Lemma}
\newtheorem{proof}{Proof}
\newtheorem{proposition}{\bf Proposition}
\usepackage{geometry}
\usepackage{color}
\definecolor{myc1}{rgb}{0,0,0}

\begin{document}

\title{ 
Fluid Antenna Relay (FAR)-assisted Communication with Hybrid Relaying Scheme Selection
}

\author{
\IEEEauthorblockN{
Ruopeng Xu$\IEEEauthorrefmark{1}$,
Songling Zhang$\IEEEauthorrefmark{1}$,
Zhaohui Yang$\IEEEauthorrefmark{1}$,
Mingzhe Chen$\IEEEauthorrefmark{2}$,
Zhaoyang Zhang$\IEEEauthorrefmark{1}$,
Kai-Kit Wong$\IEEEauthorrefmark{3}$
}
	\IEEEauthorblockA{
			$\IEEEauthorrefmark{1}$College of Information Science and Electronic Engineering, Zhejiang University, Hangzhou, China\\
            $\IEEEauthorrefmark{2}$Department of Electrical and Computer Engineering, University of Miami \\
            $\IEEEauthorrefmark{3}$Department of Electronic and Electrical Engineering, University College London, U.K.\\
          	E-mails:
(ruopengxu, sl-zhang, yang\_zhaohui, ning\_ming)@zju.edu.cn, mingzhe.chen@miami.edu, kai-kit.wong@ucl.ac.uk
		}
\vspace{-3em}
}
\maketitle

\maketitle

\begin{abstract}
In this paper, we investigate a fluid antenna relay (FAR)-assisted communication system with hybrid relaying scheme selection. By leveraging statistical channel state information (CSI) and distribution characteristics of fluid antenna system (FAS), we approximate the outage probability (OP) with different relaying schemes utilizing a Gaussian copula-based method. Each relay node follows the OP-minimized principle to choose the forwarding schemes. To reduce self-interference and avoid multi-user interference, half-duplex relays and frequency division multiple access schemes are considered, respectively. On this basis, we formulate a sum-rate maximization problem to mitigate the rate loss introduced by the half-duplex mode. To solve this problem, we first transform the original nonconvex problem into a power control optimization problem by obtaining the closed form of bandwidth allocation and substituting it into the original problem. Then, we solve the power control optimization problem with a low complexity method. Simulation results verify the effectiveness of our proposed algorithm to improve the sum rate of the system.
\end{abstract}

\begin{IEEEkeywords}
Fluid antenna system (FAS), fluid antenna relay (FAR), hybrid relaying schemes selection, outage probability
\end{IEEEkeywords}
\IEEEpeerreviewmaketitle

\section{Introduction}\label{Introduction}
The rapid evolution towards the sixth-generation (6G) wireless communication systems presents unprecedented challenges and demands, particularly regarding network coverage, capacity in high-frequency bands, and energy efficiency across the communication system\cite{8869705}. Given that wireless communications in 6G will extensively utilize high-frequency spectrum, severe signal attenuation becomes critical hurdles that traditional direct transmission cannot overcome, making cooperative relaying an critical technology for signal transmissions. To maximize the performance gain of these essential relay networks, we introduce the fluid antenna system (FAS) proposed in\cite{9264694}, which is capable of dynamically selecting the optimal spatial location to receive or transmit a signal, and have proposed the concept of fluid antenna relay (FAR) in \cite{10615841}.

Several works have also investigated the combination of the FAS in relay networks, improving the performance of the communication systems\cite{10167904,aka2025power,11216397}. However, existing works primarily consider the single relaying scheme, i.e., using either amplify-and-forward (AF) or decode-and-forward (DF) for relaying, neglecting the need for different relaying schemes in different communication conditions. Therefore, we propose a hybrid relaying scheme selection strategy, which combines the statistical channel state information (CSI) with the distribution characteristics of the FAS channels to select the appropriate relay forwarding scheme at the relay nodes. On the other hand, a full-duplex relaying system is susceptible to interference originating from self-interference, other relays, or other users' signals, and a half-duplex (HD) relaying system can eliminate self-interference though, it reduces the transmission rate of the system. To address this issue, we consider a HD relaying system and utilize orthogonal multiple access (OMA) for interference cancellation, while formulating an optimization problem to maximize the sum rate of the system.

The key contributions of this paper include:
\begin{itemize}
    \item We investigate an FAR-assisted multi-user uplink communication system, where FARs integrate the FAS and relay, introducing an extra degree of freedom and improving the performance of the system. Each FAR utilizes hybrid relaying scheme selection strategy that follows the outage probability (OP)-minimized principle. This strategy leverages the statistical CSI and the distribution characteristics of the FAS channels to select between the AF and DF schemes at the relay node to minimize the OP of the whole transmission.
    \item We formulate a sum-rate maximization problem for the proposed system, jointly optimizing the bandwidth allocation and transmit power at users and FAR nodes. To solve this optimization problem, we first obtain the closed-form of the optimal bandwidth allocation. Then, by substituting the optimal solution of the bandwidth allocation into the original problem, we equivalently transform it into transmit power optimization problem, which can be solved by a low complexity algorithm.
    \item We introduce a Gaussian copula-based method to approximate the OP value when using AF or DF under different transmit powers and outage thresholds. Besides, simulation results are provided to demonstrate the impact of the number of users, the number of ports and the maximum transmit power of FARs, validating that our proposed scheme outperforms all the benchmark schemes.
\end{itemize}

\section{System Model}

\begin{figure}[t]
\centering
\includegraphics[width=1\linewidth]{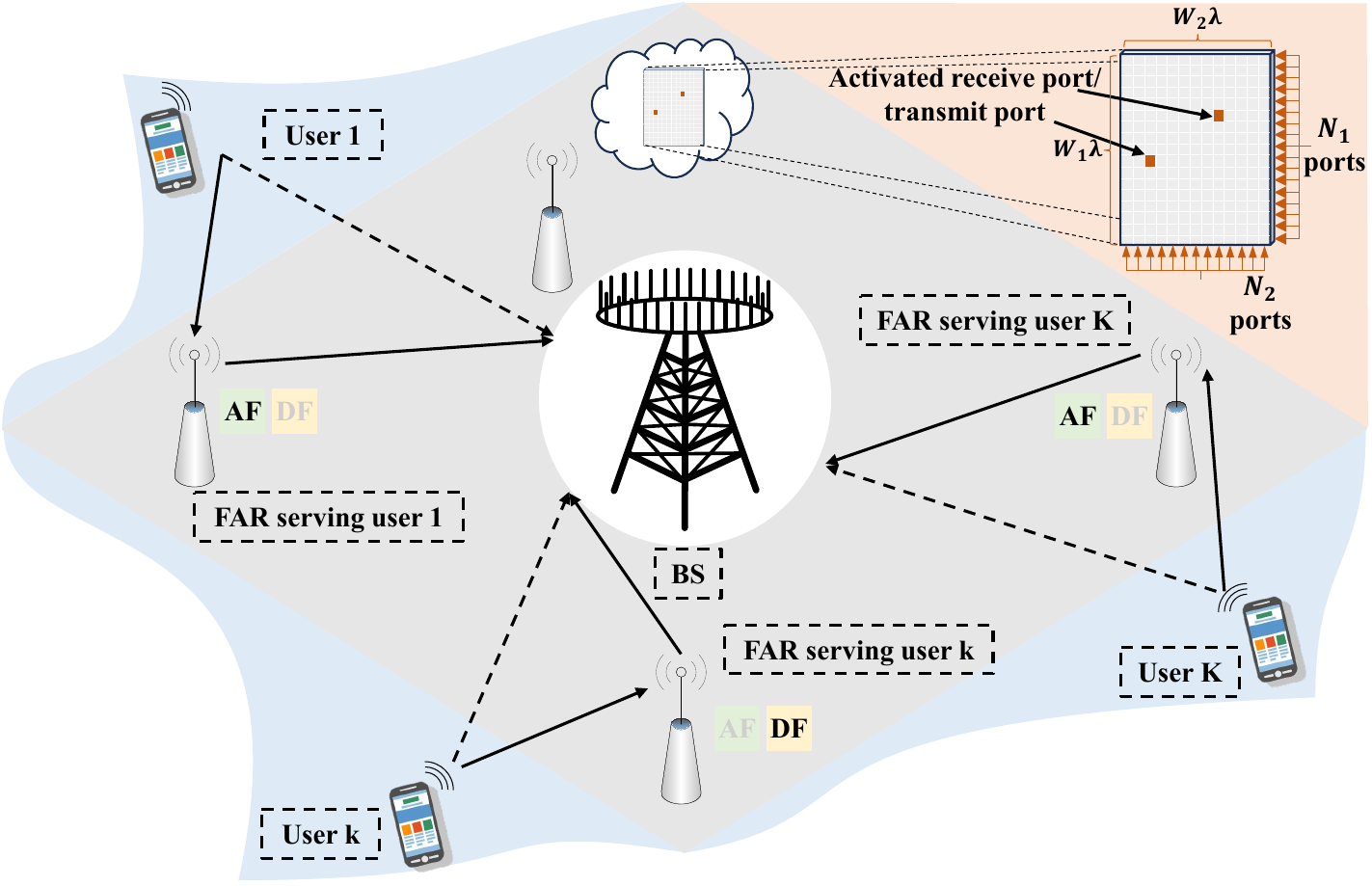}
\caption{System model of the proposed communication system.} 
\label{SystemModel}
\vspace{-0.6cm}
\end{figure}

As illustrated in Fig.~\ref{SystemModel}, we consider an FAR-assisted uplink communication system, which consists of $K$ users, $K$ FARs, and one base station (BS), where the users and BS are equipped with traditional antenna system (TAS) including a single position-fixed antenna, and each FAR deploys a single two-dimensional (2D) fluid antenna (FA). We assume that all users are far away from the BS such that the line-of-sight (LoS) links between the users and the BS are weak. To maintain the quality of the communication, at the first time slot, each user broadcasts its message to both the BS and a nearby FAR. Then, at the second time slot, the FAR transmits the processed signal to the BS using either the AF or DF scheme. By utilizing the maximal ratio combining (MRC) method, the BS combines the received signals from user $k$ and the corresponding FAR (the $k$-th FAR) serving user $k$. In particular, we consider HD relays to reduce the self-interference and frequency division multiples access (FDMA) employed in the BS to avoid multi-user interference and better perform the MRC method.
\subsection{Fluid Antenna Relay Model}
It is assumed that the 2D FA includes only one radio frequency (RF) chain and $N=N_1\times N_2$ preset ports, where the $N_i$ ports are uniformly distributed along a linear space of length $W_i\lambda$ for $i\in\{1,2\}$, taking up a grid surface size of $W=W_1\times W_2\lambda^2$ with $\lambda$ being the wave length. For simplicity of notation, a mapping function $\mathcal{F}:\mathbb{R}^2 \rightarrow \mathbb{R}$ is introduced, e.g., for the $(n_1,n_2)$-th port, we can note it as the $l$-th port,
\begin{equation}
    l=\mathcal{F}(n_1,n_2)=(n_1-1)N_2+ n_2.
\end{equation}
Based on this, for the $k$-th FAR, we denote the channel gain at the $l$-th port by $\sqrt{\alpha_k^{\mathrm{UR}}}h_k^l$, where $\alpha_k^{\mathrm{UR}}$ is the large-scale fading component related to the distance
between the user $k$ and the $k$-th FAR, and $h_k^{l}$ is the normalized channel gain following a circularly symmetric complex Gaussian distribution with zero mean and unit variance\cite{11193779}, i.e., $h_k^{l} \sim \mathcal{CN} (0,1)$. Furthermore, the square of the amplitude of $h_k^l$, $|h_k^l|^2$, follows the exponential distribution with the parameter $1$.
 
The received signal at the $k$-th FAR can be obtained as
\begin{equation}\label{URlink}
    y_k^{\mathrm{UR}} = \sqrt{p_{k}^{\mathrm{U}}} \sqrt{\alpha_k^{\mathrm{UR}}} h_{k}^{\mathrm{UR}} x_k + n_{k},
\end{equation}
where ${p_{k}^{\mathrm{U}}}$ is transmit power of user $k$, $h_{k}^{\mathrm{UR}}$ is channel gain of the optimal port which maximizes signal-to-noise ratio ($\mathrm{SNR}$) of the received signal, $x_k \sim \mathcal{CN} (0,1)$ is the transmitted signal of user $k$, and $n_k$ $\sim$ $\mathcal{CN}$($0$,$\sigma_k^2$) is the additive white Gaussian noise (AWGN). Mathematically, $h_{k}^{\mathrm{UR}}$ can be obtained as
\begin{equation}\label{FAS channel gain}
    |h_k^{\mathrm{UR}}|^2 = \mathop{\mathrm{max}} \{|h_k^1|^2,|h_k^2|^2,\dots,|h_k^N|^2\},
\end{equation}
where $\{h_k^l\}_{\forall l}$ are correlated with the correlation matrix $\mathbf{J}$,
\begin{align}\label{SpatialCorrMatrix}
    \begin{aligned}
      \mathbf{J} 
      =\begin{bmatrix}
      1 & J_{1,2} & \dots & J_{1,N} \\
      J_{2,1} & 1 & \dots & J_{2,N} \\
      \vdots & \vdots & &  \vdots\\
      J_{N,1} & J_{N,2} & \dots & 1 \\
      \end{bmatrix},
    \end{aligned}
\end{align}
where the entry $J_{k,l}$ is the spatial correlation between the $k$-th port and the $l$-th port. In specific, the spatial correlation between the $(n_1,n_2)$-th port and $(\tilde{n}_1,\tilde{n}_2)$-th port can be described as\cite{10303274}
\begin{align}
\nonumber
    &J_{(n_1,n_2),(\tilde{n}_1,\tilde{n}_2)}\\
    &=j_0\left(2\pi\sqrt{\left(\frac{|n_1-\tilde{n}_1|}{N_1-1}W_1\right)^2+\left(\frac{|n_2-\tilde{n}_2|}{N_2-1}W_2\right)^2}\right),
\end{align}
where $j_0(\cdot)$ is the spherical Bessel function of the first kind.

\subsection{Communication Model}
During one specific transmission, user $k$ first transmits the message both to the BS and the $k$-th FAR. For the BS, the received signal of this time slot can be given as
\begin{equation}
    y_k^{\mathrm{UB}} = \sqrt{p_k^{\mathrm{U}}} \sqrt{\alpha_k^{\mathrm{UB}}}h_k^{\mathrm{UB}}x_k+n_0,
\end{equation}
where $\alpha_k^{\mathrm{UB}}$ is the large-scale component related to the distance between the $k$-th FAR and BS, $h_k^{\mathrm{UB}} \sim \mathcal{CN}(0,1)$ is the normalized channel gain from user $k$ to BS, and $n_0 \sim \mathcal{CN}(0,\sigma^2_0)$ is the AWGN at the BS.

If the $k$-th FAR uses AF scheme for relaying, the received signal at the BS can be given as
\begin{equation}\label{AF RB}
    y_k^{\mathrm{AF}} = \sqrt{p_{k}^{\mathrm{R}}} \sqrt{\alpha_k^{\mathrm{RB}}} h_{k}^{\mathrm{RB}} \left(A_ky_k^{\mathrm{UR}}\right) + n_0,
\end{equation}
where 
$\alpha_k^{\mathrm{RB}}$ is the large-scale component related to the distance between the FAR and the BS, $h_{k}^{\mathrm{RB}}$ is the normalized channel gain of the optimal port sending signals from the $k$-th FAR to BS, and $A_k = \sqrt{\frac{1}{p_{k}^{\mathrm{U}}\alpha_k^{\mathrm{UR}}|h_{k}^{\mathrm{UR}}|^2+\sigma_k^2}}$ is the amplifying gain. 

Under the considered model, the BS combines the signals $y_k^{\mathrm{UB}}$ and $y_k^{\mathrm{AF}}$ with methods such as MRC, and the overall SNR of the received signal at the BS can be given as\cite{5545638} 
\begin{equation}\label{AF SNR}
    \Gamma_k^{\mathrm{AF}} 
    = p_k^\mathrm{U}{\gamma}_k^{\mathrm{UB}} + \frac{p_{k}^{\mathrm{U}}\gamma_k^{\mathrm{UR}}p_{k}^{\mathrm{R}}{\gamma}_k^{\mathrm{RB}}}{p_{k}^{\mathrm{R}}{\gamma}_k^{\mathrm{RB}}+p_{k}^{\mathrm{U}}\gamma_k^{\mathrm{UR}}+1},
\end{equation}
where $\gamma_k^\mathrm{UB} = \frac{\alpha_k^{\mathrm{UB}}|h_{k}^{\mathrm{UB}}|^2}{\sigma_0^2}$, $\gamma_k^{\mathrm{UR}} = \frac{ \alpha_k^{\mathrm{UR}} |h_{k}^{\mathrm{UR}}|^2}{\sigma_k^2}$ and $\gamma_k^\mathrm{RB} = \frac{\alpha_k^{\mathrm{RB}}|h_{k}^{\mathrm{RB}}|^2}{\sigma_0^2}$ are SNRs normalized by the corresponding transmit powers. Moreover, we define $\Gamma_k^{\mathrm{UB}}=p_k^{\mathrm{U}}\gamma_k^\mathrm{UB}$, $\Gamma_k^{\mathrm{UR}}=p_k^{\mathrm{U}}\gamma_k^\mathrm{UR}$ and $\Gamma_k^\mathrm{RB}=p_{k}^{\mathrm{R}}\gamma_k^{\mathrm{RB}}$ as the corresponding SNRs.

On the other hand, if adopt the DF scheme, the FAR first decodes the received signal and transmits it to the BS. As is widely adopted in relay communication networks, such as \cite{1362898,5397895,6138259}, we assume that the signal can be decoded without any error. 
Then, the received signal at the BS and its corresponding SNR can be respectively given as
\begin{align}
    y_k^{\mathrm{DF}} &= \sqrt{p_{k}^{\mathrm{R}}} \sqrt{\alpha_k^{\mathrm{RB}}} h_{k}^{\mathrm{RB}} x_k + n_0,
\end{align}
and
\begin{align}
    \Gamma_k^{\mathrm{DF}} 
    &= \mathrm{min}\{{\Gamma}_k^{\mathrm{UB}}+ {\Gamma}_k^\mathrm{RB},\Gamma_k^{\mathrm{UR}}\}.
\end{align}

As a result, the achievable rate of user $k$ can be given as 
\begin{equation}
    r_k = \frac{1}{2}b_k \mathrm{log}_2(1+\Gamma_k),
\end{equation}
where $b_k$ is the bandwidth allocated to user $k$, $\Gamma_k = \mu_k \Gamma_k^\mathrm{AF} + (1-\mu_k) \Gamma_k^{\mathrm{DF}}$, and $\mu_k=\{1,0\}$.


\subsection{Problem Formulation}
We assume that optimal port of FAS can be obtained using methods such as \cite{10615841,aka2025power}. 
Mathematically, 
the sum-rate maximization problem can be formulated as
\begin{subequations}\label{sys1max0}
    \begin{align} 
       \mathop{\max}_{\mathbf{p}^{\mathrm{U}}, \mathbf{p}^{\mathrm{R}}, \mathbf{b},\bm{\mu}} \quad &\frac{1}{2}\sum_{k=1}^K b_k \mathrm{log}_2(1+\mu_k \Gamma_k^\mathrm{AF} + (1-\mu_k) \Gamma_k^{\mathrm{DF}}) ,\tag{\ref{sys1max0}}\\
         \textrm{s.t.} \qquad 
            & \sum_{k=1}^K b_k \leq B,\\
         &r_k \geq R_{k,\mathrm{min}},\forall k \in \mathcal{K}, \\
   & 0\leq p_k^{\mathrm{U}}\leq P_{k}^{\mathrm{U},\mathrm{max}},\forall k \in \mathcal{K},\\
   & 0\leq p_k^{\mathrm{R}}\leq P_{k}^{\mathrm{R},\mathrm{max}},\forall k \in \mathcal{K},\\
   & \mu_k = \{0,1\}, \forall k\in \mathcal{K},\\
   & b_k \geq 0, \forall k \in \mathcal{K},
    \end{align}
\end{subequations}
where $\mathbf{p}^{\mathrm{U}}=[p_1^{\mathrm{U}},\dots,p_K^{\mathrm{U}}]$, $\mathbf{p}^{\mathrm{R}}=[p_1^{\mathrm{R}},\dots,p_K^{\mathrm{R}}]$ and $\mathbf{b}=[b_1,\dots,b_K]$, $\bm{\mu}=[\mu_1,\dots,\mu_K]$ is the relaying scheme selection vector, $\mathcal{K}=\{1,\dots,K\}$ is a number set, $B$ is the total bandwidth of the system, $R_{k,\mathrm{min}}$ is the minimum achievable rate of user $k$, and $P_{k}^{\mathrm{U},\mathrm{max}}$ and $P_{k}^{\mathrm{R},\mathrm{max}}$ are the maximum transmit powers of user $k$ and the $k$-th FAR, respectively.


\section{Algorithm Design}
In this section, we first propose an OP analysis-based for each FAR node to choose the relaying scheme, based on which the FAR node minimizes the OP
. Then, to solve problem \eqref{sys1max0}, we obtain the closed form of the optimal bandwidth allocation, and by substituting it into the original problem, we reformulate problem \eqref{sys1max0} as a power control problem. Finally, we propose a low complexity method to solve this power control problem.

\subsection{OP analysis-based method for relaying scheme selection}\label{Threshold}
Given the fact that FAR nodes cannot always obtain CSI of links from users to BS and from FARs to BS, without loss of generality, we assume that the $k$-th FAR node has the average SNRs from the user $k$ to the BS and from itself to the BS, denoted as $\bar{\Gamma}_k^{\mathrm{UB}}= p_k^{\mathrm{U}} \bar{\gamma}_{k}^{\mathrm{UB}}$ and $\bar{\Gamma}_k^{\mathrm{RB}}= p_k^{\mathrm{R}}\bar{\gamma}_{k}^{\mathrm{RB}}$, respectively. 

Then, if use the AF scheme for relaying, the $k$-th FAR can estimate the OP of this transmission as
\begin{align}\label{AF OP}
\nonumber
    \mathbb{P}_{\mathrm{AF}}(\xi) &= \mathbb{P}\left(\frac{1}{2}\mathrm{log}_2(1+\Gamma_k^{\mathrm{AF}})<\xi\right)\\
\nonumber
    &=\mathbb{P}\left(p_k^\mathrm{U}\bar{\gamma}_k^{\mathrm{UB}} + \frac{p_{k}^{\mathrm{U}}\gamma_k^{\mathrm{UR}}p_{k}^{\mathrm{R}}\bar{\gamma}^{\mathrm{RB}}}{p_{k}^{\mathrm{U}}\gamma_k^{\mathrm{UR}}+p_{k}^{\mathrm{R}}\bar{\gamma}^{\mathrm{RB}}+1}<C_\mathrm{th}\right)\\
    &\overset{(\mathrm{a})}{=}
    \left\{
        \begin{aligned}
     &1, \ \ \ \ \ \ \ \ \ \ p_k^{\mathrm{U}}\bar{\gamma}_k^{\mathrm{UB}}+p_k^{\mathrm{R}}\bar{\gamma}_k^{\mathrm{RB}} < C_\mathrm{th}\\
       & F_{|h_k^{\mathrm{UR}}|^2}(\xi_{\mathrm{AF}}(p_k^{\mathrm{U}},p_k^{\mathrm{R}})),\ \ \ \ \ \ \ \ \mathrm{otherwise}
    \end{aligned}
    \right.,  
\end{align}
where $C_\mathrm{th}=2^{2\xi}-1$ is the outage threshold, $\xi_{\mathrm{AF}}(p_k^{\mathrm{U}},p_k^{\mathrm{R}})$ is a function determined by the variables $p_k^{\mathrm{U}}$ and $p_k^{\mathrm{R}}$, given as
\begin{equation}
    \xi_{\mathrm{AF}}(P_k^{\mathrm{U}},p_k^{\mathrm{R}})=\frac{\sigma_k^2(p_{k}^{\mathrm{R}}\bar{\gamma}_k^{\mathrm{RB}}+1)(C_\mathrm{th}-p_k^\mathrm{U}\bar{\gamma}_k^{\mathrm{UB}})}{\alpha_k^{\mathrm{UR}} p_k^\mathrm{U}[p_k^\mathrm{U}\bar{\gamma}_k^{\mathrm{UB}}+p_{k}^{\mathrm{R}}\bar{\gamma}_k^{\mathrm{RB}}-C_\mathrm{th}]},
\end{equation}
and $F_{|h_k^{\mathrm{UR}}|^2}$($\cdot$) is the cumulative density function (CDF) of the random variable $|h_k^{\mathrm{UR}}|^2$. The equality $({\mathrm{a}})$ holds mathematically since only the term $\gamma_k^{\mathrm{UR}}$ contains the random variable $|h_k^{\mathrm{UR}}|^2$, $\mathbb{P}_{\mathrm{AF}}(\xi)$ can be viewed as the CDF of $|h_k^{\mathrm{UR}}|^2$, and when $p_k^{\mathrm{U}}\bar{\gamma}_k^{\mathrm{UB}}+p_k^{\mathrm{R}}\bar{\gamma}_k^{\mathrm{RB}} < C_\mathrm{th}$, the CDF can be obtained as
\begin{equation}
    \mathbb{P}\left(|h_k^{\mathrm{UR}}|^2>\xi_{\mathrm{AF}}(p_k^{\mathrm{U}},p_k^{\mathrm{R}})\right)=1,
\end{equation}
where $\xi_{\mathrm{AF}}(p_k^{\mathrm{U}},p_k^{\mathrm{R}})$ is a negative number.

Similarly, if the FAR node performs the DF scheme, the OP of this transmission can be given as 
\begin{align}\label{DF OP}
\nonumber
    \mathbb{P}_{\mathrm{DF}}(\xi) &= \mathbb{P}\left(\frac{1}{2}\mathrm{log}_2(1+\Gamma_k^{\mathrm{DF}})<\xi\right)\\
    \nonumber
    &=\mathbb{P}\left(\mathrm{min}\{p_k^{\mathrm{U}}\bar{\gamma}^{\mathrm{UB}}+ p_k^{\mathrm{R}}\bar{\gamma}^\mathrm{RB},p_{k}^{\mathrm{U}}\gamma_k^{\mathrm{UR}}\}<C_{\mathrm{th}}\right)\\
    &\overset{(\mathrm{b})}{=}
    \left\{
        \begin{aligned}
     &1, \ \ \ \ \ \ \ \ \ \ p_k^{\mathrm{U}}\bar{\gamma}_k^{\mathrm{UB}}+p_k^{\mathrm{R}}\bar{\gamma}_k^{\mathrm{RB}} < C_{\mathrm{th}}\\
       & F_{|h_k^{\mathrm{UR}}|^2}(\xi_{\mathrm{DF}}(p_k^{\mathrm{U}},p_k^{\mathrm{R}})),\ \ \ \ \ \ \ \ \mathrm{otherwise}
    \end{aligned}
    \right.,
\end{align}
where 
\begin{equation}
    \xi_{\mathrm{DF}}(p_k^{\mathrm{U}},p_k^{\mathrm{R}})=\frac{\sigma^2_k(2^{2\xi}-1)}{\alpha_k^{\mathrm{UR}} p_k^{\mathrm{U}}}.
\end{equation}
The equality $(\mathrm{b})$ holds because the FAR node can first judge whether $p_k^{\mathrm{U}}\bar{\gamma}_k^{\mathrm{UB}}+p_k^{\mathrm{R}}\bar{\gamma}_k^{\mathrm{RB}} \geq C_{\mathrm{th}}$, where when the inequality does not hold, $\mathbb{P}_{\mathrm{DF}}(\xi)$ always equals $1$, and when inequality holds, we can obtain the equality $(\mathrm{b})$. 



Based on the marginal distributions of $\{|h_k^l|^2\}_{\forall l}$, we introduce copula theory\cite{10.5555/1952073} to acquire the CDF of $|h_k^{\mathrm{UR}}|^2$, and specifically exploit the Guassian copula function\cite{10678877} to approximate its numerical value. Mathematically, the CDF of $|h_k^{\mathrm{UR}}|^2$ can be presented as
\begin{equation}\label{GaussionCopula}
F_{|h_k^{\mathrm{UR}}|^2}\left(x\right)=\Phi_{\mathbf{J}}\left(\phi^{-1}\left(F_{\left|h_k^{1}\right|^2}\left(x\right)\right),\ldots,\phi^{-1}\left(F_{\left|h_{k}^N\right|^2}\left(x\right)\right)\right),
\end{equation}
where $\Phi_\mathbf{J}(\cdot)$ is the joint CDF of the multivariate normal distribution with zero mean vector and correlation matrix $\mathbf{J}$, $\phi^{-1}(\cdot)$ is the quantile function of the standard normal distribution, i.e., $\phi^{-1}(x)=\sqrt{2}\ \mathrm{erf}^{-1}(2x-1)$, in which $\mathrm{erf}^{-1}$ is the inverse function of error function $\mathrm{erf}(x) = \frac{2}{\sqrt{\pi}}\int_0^{x}e^{-r^2} \mathrm{d}r$. In particular, we choose the matrix $\mathbf{J}$ given in \eqref{SpatialCorrMatrix} as the correlation matrix of $\bm{\Phi}$ to control the degree of dependence between the correlated random variables. 

Based on \eqref{GaussionCopula}, Fig.~\ref{AF_OP} and Fig~\ref{DF_OP} illustrate the outage probability versus different transmit powers and outage thresholds with the use of AF and DF, respectively. The OPs of both schemes are close to $1$ when the transmit powers are small or the threshold is high. With a determined outage threshold, OP gradually decreases with the transmit power of the user or of the FAR increasing. Moreover, to more intuitively demonstrate the numerical relationship of OP between using AF and DF for relaying, we develop Fig.~\ref{muk} to show the value of the binary variable $\mu_k$, i.e., the $k$-entry in the relaying scheme selection vector $\bm{\mu}$, which can be mathematically given as 
\begin{align}\label{OP-min Principle}
    \mu_k &= \varepsilon\left(F_{|h_k^{\mathrm{UR}}|^2}\left(\xi_{\mathrm{DF}}\right)-F_{|h_k^{\mathrm{UR}}|^2}\left(\xi_{\mathrm{AF}}\right)\right),
\end{align}
where $\varepsilon(\cdot)$ is the Heaviside function, specified as $\varepsilon(a)=1,\ \mathrm{when}\ a \geq 0$ and $\varepsilon(a)=0,\ \mathrm{when}\ a<0$. As shown in Fig.~\ref{muk}, in the simulated range of $C_{\mathrm{th}}$, OP with DF is lower than with AF when the transmit powers are relatively small, and to obtain a lower OP, the $k$-th FAR chooses DF for relaying under this condition, i.e., $\mu_k=0$. As the transmit powers increase, OP with using DF become higher than that with AF, and the $k$-th FAR determines $\mu_k=1$. In particular, $\mu_k = \mathrm{NaN}$ means whether using AF or DF, OP is always $1$, i.e., when $p_k^{\mathrm{U}}\bar{\gamma}_k^{\mathrm{UB}}+p_k^{\mathrm{R}}\bar{\gamma}_k^{\mathrm{RB}} < C_\mathrm{th}$, $\mathbb{P}(\xi)=1$, which is aligned with the outcomes of \eqref{AF OP} and \eqref{DF OP}.

\begin{figure*}[t]
    \centering
    \vspace{-2em}
    \subfigure[]{
        \begin{minipage}{0.33\textwidth}
            \centering
            \includegraphics[width=1\textwidth]{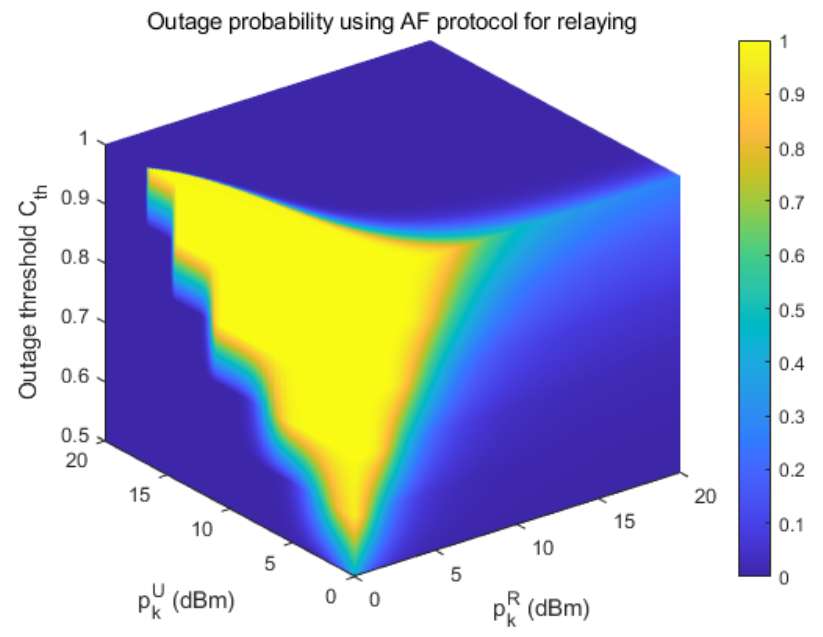}
            \label{AF_OP}
    \end{minipage}}
    \hspace{-5mm}
    \subfigure[]{
        \begin{minipage}{0.33\textwidth}
            \centering
            \includegraphics[width=1\textwidth]{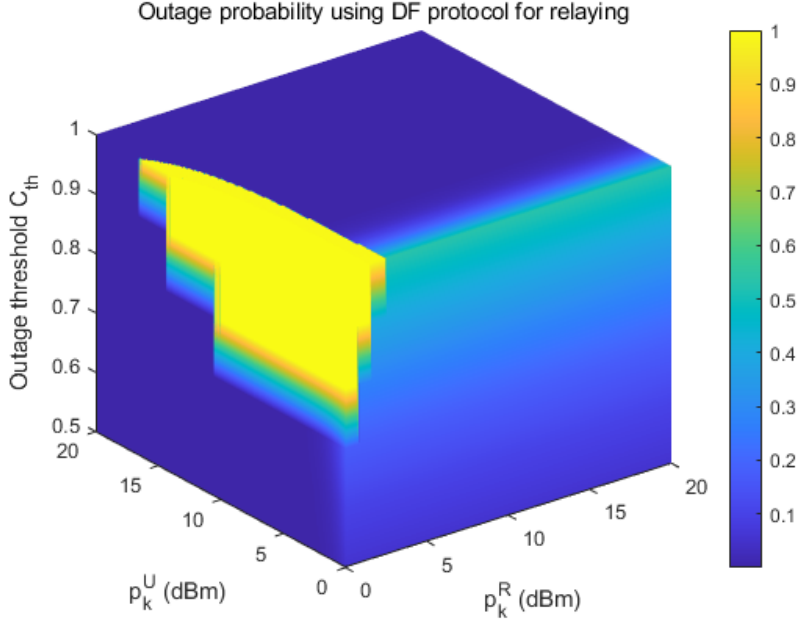}
            \label{DF_OP}	
    \end{minipage}}
        \hspace{-5mm}
        \subfigure[]{
            \begin{minipage}{0.33\textwidth}
                \centering
                \includegraphics[width=1\textwidth]{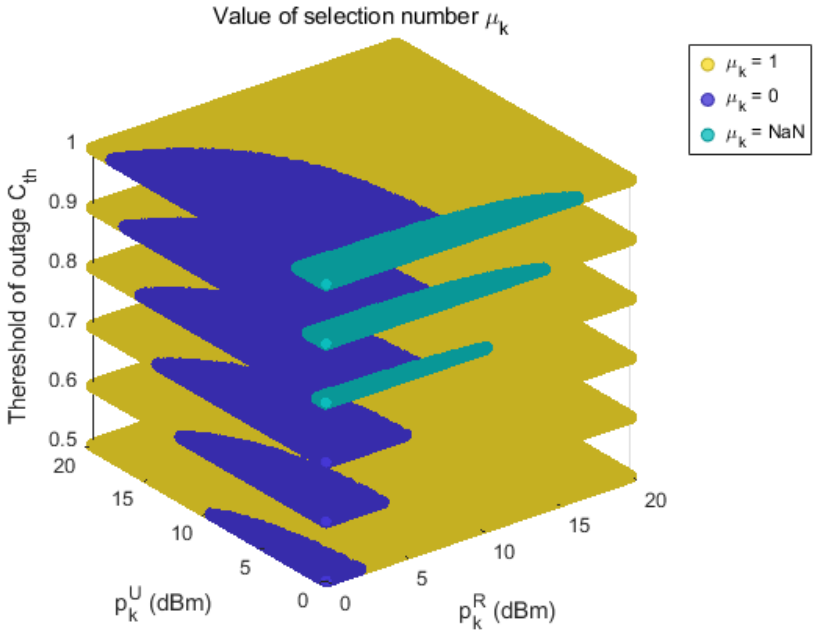}
                \label{muk}	
        \end{minipage}}
    \vspace{-0.6em}
    \caption{(a) Outage probability using AF scheme for relaying, (b) Outage probability using DF scheme for relaying, and (c) Value of selection number $\mu_k$.}
    \label{fig.} 
    \vspace{-1em}
\end{figure*}

\subsection{Sum-rate Maximization}\label{sumrate max}
After determining the OP minimized principle for relaying scheme selection, we solve the problem \eqref{sys1max0} as follows.

\subsubsection{Bandwidth allocation optimization}
With given $\mathbf{p}^{\mathrm{U}}$ and $\mathbf{p}^{\mathrm{R}}$, selection vector $\bm{\mu}$ can be determined based on $\eqref{OP-min Principle}$ and we can formulate the bandwidth allocation sub-problem as
\begin{subequations}\label{subBand}
    \begin{align} 
       \mathop{\max}_{\mathbf{b}} \quad &\sum_{k=1}^K \frac{1}{2}b_k \mathrm{log}_2(1+\Gamma_k) ,\tag{\ref{subBand}}\\
         \textrm{s.t.} \qquad 
            & \sum_{k=1}^K b_k \leq B,\\
         &\frac{1}{2}b_k \mathrm{log}_2(1+\Gamma_k) \geq R_{k,\mathrm{min}},\forall k \in \mathcal{K}, \\
   & b_k \geq 0, \forall k \in \mathcal{K},
    \end{align}
\end{subequations}
which is a linear optimization problem. Hence, to maximize the sum rate of the system, additional bandwidth is supposed to be allocated to the user with the highest SNR. Assume user $n$ holds the highest SNR, and $n$ can be presented as
\begin{equation}\label{n}
    n = \mathrm{arg}\ \mathop{\mathrm{max}}_{k}\ \{\Gamma_k\}.
\end{equation}
Then, the closed form of the optimal bandwidth allocation vector $\mathbf{b}$ can be given as
\begin{equation}\label{opt B}
b_k^*=\left\{
    \begin{aligned}
        &\frac{2R_{k,\mathrm{min}}}{\mathrm{log}_2(1+\Gamma_k)}, \forall k\in \mathcal{K}\ ,  k\neq n\\
        &B - \sum_{k=1,k\neq n}^K \frac{2R_{k,\mathrm{min}}}{\mathrm{log}_2(1+\Gamma_k)}, k=n
    \end{aligned}
    \right..
\end{equation}
\subsubsection{Power control optimization}
Based on the OP-minimized principle, $\mu_k$ is determined with given $p_k^{\mathrm{U}}$ and $p_k^{\mathrm{R}}$. Hence, we only need to optimize $\mathbf{p}^{\mathrm{U}}$ and $\mathbf{p}^{\mathrm{R}}$ after we obtain the optimal $\mathbf{b}$. Moreover, by substituting \eqref{opt B} into the problem \eqref{sys1max0}, we can equivalently change original problem into the power control optimization problem formulated as
\begin{subequations}\label{subP}
    \begin{align} 
       \mathop{\max}_{\mathbf{p}^{\mathrm{U}}, \mathbf{p}^{\mathrm{R}}}& \frac{1}{2}\left(B - \sum_{k=1,k\neq n}^K \frac{2R_{k,\mathrm{min}}}{\mathrm{log}_2(1+\Gamma_k)}\right) \mathrm{log}_2(1+\Gamma_n) ,\tag{\ref{subP}}\\
         \textrm{s.t.} \quad 
         & B - \sum_{k=1,k\neq n}^K \frac{2R_{k,\mathrm{min}}}{\mathrm{log}_2(1+\Gamma_k)}>0,\\
         &\frac{1}{2}(B - \sum_{k=1,k\neq n}^K \frac{2R_{k,\mathrm{min}}}{\mathrm{log}_2(1+\Gamma_k)})\mathrm{log}_2(1+\Gamma_n)\geq R_{n,\mathrm{min}},\\
          &\Gamma_k={\mu}_k\Gamma_{\mathrm{AF}}+(1-{\mu}_k)\Gamma_{\mathrm{DF}}, \forall k \in \mathcal{K},\\
         &(\ref{sys1max0}\mathrm{c}),
         \nonumber(\ref{sys1max0}\mathrm{d}),\eqref{OP-min Principle},\eqref{n},
    \end{align}
\end{subequations}
where constraints $(\ref{subP}\mathrm{a})$ and $(\ref{subP}\mathrm{b})$ can be omitted, since they should always be satisfied, or there is no solution to this problem. The power control problem is not convex, because the objective function, constraints \eqref{OP-min Principle}, \eqref{n} and $(\ref{subP}\mathrm{c})$ are not convex. Due to the fact that the selection vector for relaying scheme varies when either $\mathbf{p}^\mathrm{U}$ or $\mathbf{p}^\mathrm{R}$ changes, the problem is complicated to solve.

To solve this problem, we first recall that when $p_k^{\mathrm{U}}\bar{\gamma}_k^{\mathrm{UB}}+p_k^{\mathrm{R}}\bar{\gamma}_k^{\mathrm{RB}} < C_{\mathrm{th}}$, the OP always equals 1, and thus neither of the schemes can be chosen. Without loss of generality, we can reformulate problem \eqref{subP} as the following problem
\begin{subequations}\label{subP2}
    \begin{align} 
       \mathop{\max}_{\mathbf{p}^{\mathrm{U}}, \mathbf{p}^{\mathrm{R}}} \quad & \frac{1}{2}\left(B - \sum_{k=1,k\neq n}^K \frac{2R_{k,\mathrm{min}}}{\mathrm{log}_2(1+\Gamma_k)}\right) \mathrm{log}_2(1+\Gamma_n) ,\tag{\ref{subP2}}\\
         \textrm{s.t.} \qquad 
   & P_k^{\mathrm{U,min}}\leq p_k^{\mathrm{U}}\leq P_{k}^{\mathrm{U,max}},\forall k \in \mathcal{K},\\
   & P_k^{\mathrm{R,min}}\leq p_k^{\mathrm{R}}\leq P_{k}^{\mathrm{R,max}},\forall k \in \mathcal{K},\\
   &P_k^{\mathrm{U,min}}\bar{\gamma}_k^{\mathrm{UB}}+P_k^{\mathrm{R,min}}\bar{\gamma}_k^{\mathrm{RB}} \geq C_{\mathrm{th}},\\
   \nonumber
   &\eqref{OP-min Principle},\eqref{n},(\ref{subP}\mathrm{c}).
    \end{align}
\end{subequations}

Based on the fact that $\Gamma_{k}^{\mathrm{AF}}$ and $\Gamma_{k}^{\mathrm{DF}}$ are both monotonically non-decreasing with respective to $p_k^\mathrm{U}$ and $p_k^{\mathrm{R}}$, we propose the following lemmas to help us solve the problem \eqref{subP2}.
\begin{lemma}\label{lemma1}
    If $\tilde{\varepsilon}(P_k^{\mathrm{U,max}},P_k^{\mathrm{R,max}})=0$, we always have
    \begin{equation}
        \tilde{\varepsilon}(P_k^{\mathrm{U,max}}-\delta_1,P_k^{\mathrm{R,max}}-\delta_2)=0,
    \end{equation}
where $0\leq\delta_1 \leq \Delta P_k^{\mathrm{U}}$, $0\leq\delta_2 \leq \Delta P_k^{\mathrm{R}}$, $\Delta P_k^{\mathrm{U}}=P_k^{\mathrm{U,max}}-P_k^{\mathrm{U,min}}$, and $\Delta P_k^{\mathrm{R}}=P_k^{\mathrm{R,max}}-P_k^{\mathrm{R,min}}$.
\end{lemma}

\begin{proof}\label{proof1}
   Since $\tilde{\varepsilon}(P_k^{\mathrm{U,max}},P_k^{\mathrm{R,max}})=0$, we have $\xi_{\mathrm{AF}}(P_k^{\mathrm{U,max}},P_k^{\mathrm{R,max}})>\xi_{\mathrm{DF}}(P_k^{\mathrm{U,max}},P_k^{\mathrm{R,max}})$, i.e.,
   \begin{equation}\label{Final 0}
      \frac{C_{\mathrm{th}}^2+C_{\mathrm{th}}}{(C_{\mathrm{th}}+1)P_k^{\mathrm{\mathrm{U,max}}}\bar{\gamma}_k^{\mathrm{UB}}+P_k^{\mathrm{\mathrm{U,max}}}\bar{\gamma}_k^{\mathrm{UB}}P_k^{\mathrm{\mathrm{R,max}}}\bar{\gamma}_k^{\mathrm{RB}}} > 1.
   \end{equation}
It is obvious that the denominator of the term on the left side of \eqref{Final 0} is monotonically increasing with respective to $p_k^\mathrm{U}$ and $p_k^{\mathrm{R}}$. Hence, if we assume that there are $\delta_1$ and $\delta_2$, making $\frac{\xi_{\mathrm{AF}}(P_k^{\mathrm{U,max}}-\delta_1,P_k^{\mathrm{R,max}}-\delta_2)}{\xi_{\mathrm{DF}}(P_k^{\mathrm{U,max}}-\delta_1,P_k^{\mathrm{R,max}}-\delta_2)} \leq 1$, we can further obtain
   \begin{align}
   \nonumber
 C_{\mathrm{th}}^2
 &+C_{\mathrm{th}}
 -(C_{\mathrm{th}}+1)(P_k^{\mathrm{\mathrm{U,max}}}
 -\delta_1)\bar{\gamma}_k^{\mathrm{UB}}\\
 &-(P_k^{\mathrm{\mathrm{U,max}}}-\delta_1)\bar{\gamma}_k^{\mathrm{UB}}(P_k^{\mathrm{\mathrm{R,max}}}-\delta_2)\bar{\gamma}_k^{\mathrm{RB}} \leq 0,
   \end{align}
   which is contrary to \eqref{Final 0}. Hence, $\frac{\xi_{\mathrm{AF}}(P_k^{\mathrm{U,max}}-\delta_1,P_k^{\mathrm{R,max}}-\delta_2)}{\xi_{\mathrm{DF}}(P_k^{\mathrm{U,max}}-\delta_1,P_k^{\mathrm{R,max}}-\delta_2)} > 1$, i.e., $\tilde{\varepsilon}(P_k^{\mathrm{U,max}}-\delta_1,P_k^{\mathrm{R,max}}-\delta_2)=0$.
\end{proof}

\begin{proposition}\label{proposition1}
    When $\tilde{\varepsilon}(P_k^{\mathrm{U,max}},P_k^{\mathrm{R,max}})=0$, optimal transmit powers can be obtained as 
    \begin{equation}
           (p_k^{\mathrm{U,*}},p_k^{\mathrm{R,*}})=
        (P_k^{\mathrm{U,max}},P_k^{\mathrm{R,max}}) .
    \end{equation}
\end{proposition}
\begin{proof}
   Based on Lemma~\ref{lemma1}, if $\tilde{\varepsilon}(P_k^{\mathrm{U,max}},P_k^{\mathrm{R,max}})=0$, we always have $\tilde{\varepsilon}(p_k^{\mathrm{U}},p_k^{\mathrm{R}})=0,\forall p_k^{\mathrm{U}},p_k^{\mathrm{R}}$. Then, since $\Gamma_{k}^{\mathrm{AF}}$ and $\Gamma_{k}^{\mathrm{DF}}$ are both monotonically non-decreasing with respective to $p_k^\mathrm{U}$ and $p_k^{\mathrm{R}}$, the optimal transmit powers of user $k$ and the $k$-th FAR are $p_k^{\mathrm{U,*}} = P_k^{\mathrm{U,max}}$ and $p_k^{\mathrm{R,*}} = P_k^{\mathrm{R,max}}$, respectively.
\end{proof}

\begin{figure*}[t]
    \centering
            \vspace{-2em}
    \subfigure[]{
        \begin{minipage}{0.33\textwidth}
            \centering
            \includegraphics[width=1\textwidth]{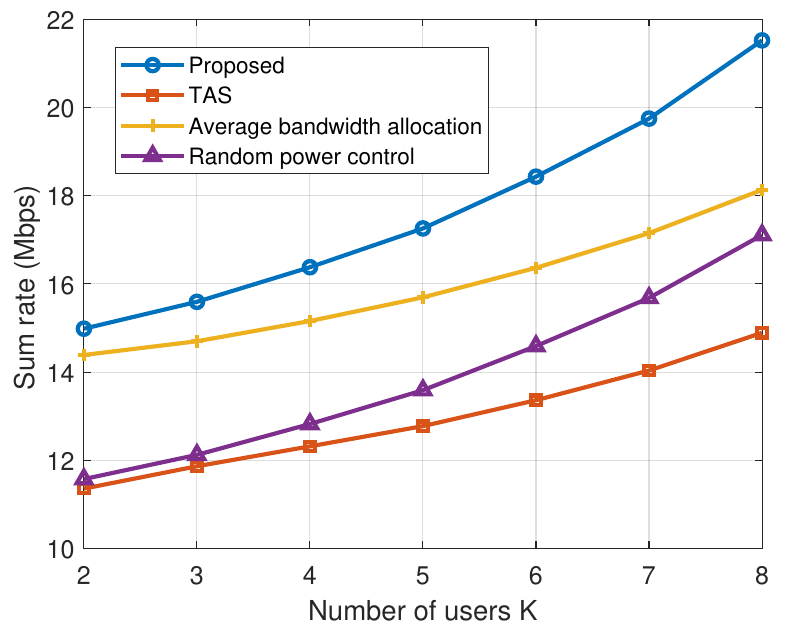}
            \label{UserNumber}
    \end{minipage}}
    \hspace{-5mm}
    \subfigure[]{
        \begin{minipage}{0.33\textwidth}
            \centering
            \includegraphics[width=1\textwidth]{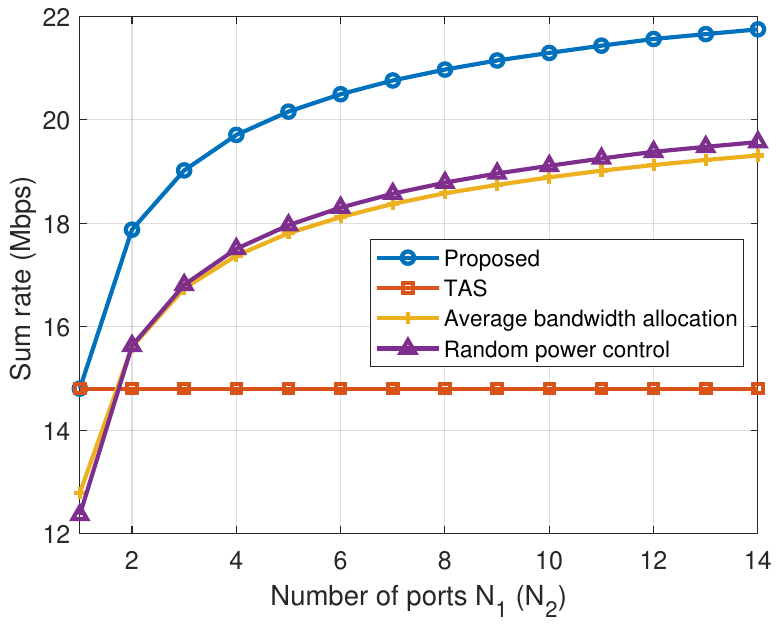}
            \label{N}	
    \end{minipage}}
        \hspace{-5mm}
        \subfigure[]{
            \begin{minipage}{0.33\textwidth}
                \centering
                \includegraphics[width=1\textwidth]{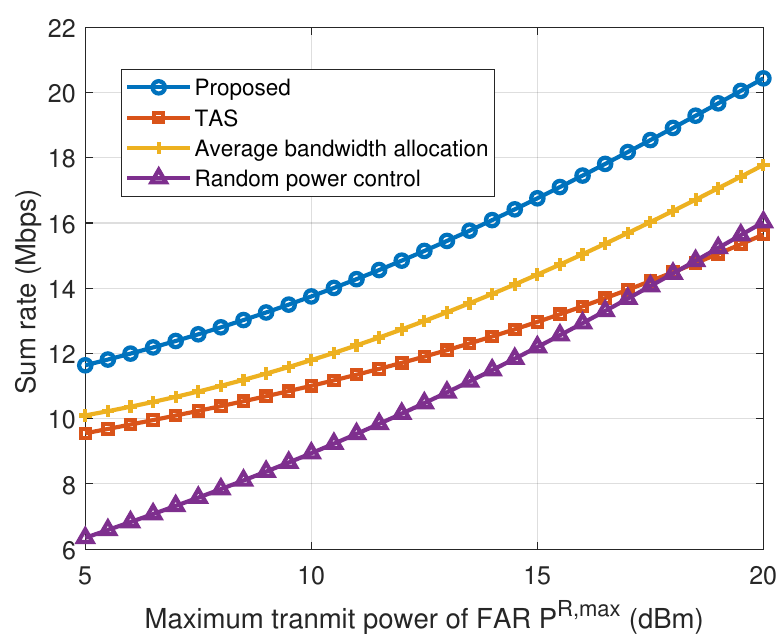}
                \label{Pr}	
        \end{minipage}}
            \vspace{-0.8em}
    \caption{Sum rate of the system versus (a) Number of users, (b) Number of ports, (c) Maximum transmit power of FAR. }
    \label{sumRate} 
    \vspace{-1.5em}
\end{figure*}

\begin{lemma}\label{lemma2}
        If $\tilde{\varepsilon}(P_k^{\mathrm{U,min}},P_k^{\mathrm{R,min}})=1$, we always have
    \begin{equation}
        \tilde{\varepsilon}(P_k^{\mathrm{U,min}}+\delta_1,P_k^{\mathrm{R,min}}+\delta_2)=1,
    \end{equation}
where $0\leq\delta_1 \leq \Delta P_k^{\mathrm{U}}$ and $0\leq\delta_2 \leq\Delta P_k^{\mathrm{R}}$.
\end{lemma}
\begin{proof}\label{proof2}
    Similar to the proof~\ref{proof1}.
\end{proof}

Based on Lemmas~\ref{lemma1} and~\ref{lemma2}, we can obtain the following proposition as
\begin{proposition}\label{proposition2}
    When $\tilde{\varepsilon}(P_k^{\mathrm{U,max}},P_k^{\mathrm{R,max}})=1$, optimal transmit powers can be obtained as 
    \begin{align}
    \nonumber
    &(p_k^{\mathrm{U,*}},p_k^{\mathrm{R,*}})=\\
    &\left\{
        \begin{aligned}
        &(\tilde{p}_k^{\mathrm{U}},\tilde{p}_k^{\mathrm{R}}),\ \mathrm{if}\ \tilde{C} \geq 1\ \mathrm{AND}\ \frac{\Gamma_k^{\mathrm{AF}}(P_k^{\mathrm{U,max}},P_k^{\mathrm{U,max}})}{\Gamma_k^{\mathrm{DF}}(\tilde{p}_k^{\mathrm{U}},\tilde{p}_k^{\mathrm{R}})}\leq1\\
            &(P_k^{\mathrm{U,max}},P_k^{\mathrm{R,max}}),\ \mathrm{otherwise}\\
        \end{aligned}
        \right.,
    \end{align}
    where $\tilde{C}=\frac{C_{\mathrm{th}}^2+C_{\mathrm{th}}}{(C_{\mathrm{th}}+1)P_k^{\mathrm{\mathrm{U,min}}}\bar{\gamma}_k^{\mathrm{UB}}+P_k^{\mathrm{\mathrm{U,min}}}\bar{\gamma}_k^{\mathrm{UB}}P_k^{\mathrm{\mathrm{R,min}}}\bar{\gamma}_k^{\mathrm{RB}}}$,
    $(\tilde{p}_k^{\mathrm{U}},\tilde{p}_k^{\mathrm{R}})$ is the optimal solution to problem \eqref{subP3}.
\end{proposition}

\begin{proof}
    When $\tilde{C} < 1$, based on Lemma~\ref{lemma2}, we always have $\tilde{\mu}_k(p_k^\mathrm{U},p_k^{\mathrm{R}})=1$. Then, according to the monotonicity of the $\Gamma_k^{\mathrm{AF}}$, transmit powers maximizing $\Gamma_k^{\mathrm{AF}}$ should satisfy $(p_k^{\mathrm{U,*}},p_k^{\mathrm{R,*}})= (P_k^{\mathrm{U,max}},P_k^{\mathrm{R,max}})$. On the other hand, if $\tilde{C} \geq 1$, it shows that $\exists$ $0\leq\delta_1 \leq\Delta P_k^{\mathrm{U}}$ and $0\leq\delta_2 \leq \Delta P_k^{\mathrm{R}}$ $s.t.$ $\tilde{\varepsilon}(P_k^{\mathrm{U,max}}-\delta_1,P_k^{\mathrm{R,max}}-\delta_2)=0$. Then, to maximize the achievable rate of user $k$, we need to compare the highest SNR when using AF scheme ($\Gamma_k^{\mathrm{AF,max}}$) and the highest SNR when using DF scheme ($\Gamma_k^{\mathrm{DF,max}}$). Based on Lemma~\ref{lemma2}, we can obtain $\Gamma_k^{\mathrm{AF,max}}$ when $p_k^{\mathrm{U}}$ and $p_k^{\mathrm{R}}$ both achieve the corresponding upper bound of its feasible region. To obtain $\Gamma_k^{\mathrm{DF,max}}$, we can formulate the following optimization problem as

    \begin{subequations}\label{subP3}
    \begin{align} 
       \mathop{\max}_{{p}_k^{\mathrm{U}}, {p}_k^{\mathrm{R}}} \quad & \mathrm{min}\{p_k^{\mathrm{U}}\gamma_k^{\mathrm{UB}}+p_k^{\mathrm{R}}\gamma_k^{\mathrm{RB}},p_k^{\mathrm{U}}\gamma_k^{\mathrm{UR}}\},\tag{\ref{subP3}}\\
         \textrm{s.t.} \qquad 
         &  \frac{C_{\mathrm{th}}^2+C_{\mathrm{th}}}{(C_{\mathrm{th}}+1)p_k^{\mathrm{\mathrm{U}}}\bar{\gamma}_k^{\mathrm{UB}}+p_k^{\mathrm{\mathrm{U}}}\bar{\gamma}_k^{\mathrm{UB}}p_k^{\mathrm{\mathrm{R}}}\bar{\gamma}_k^{\mathrm{RB}}}\geq 1, \\
   \nonumber
   &(\ref{subP2}\mathrm{a})-(\ref{subP2}\mathrm{c}),
    \end{align}
\end{subequations}
which is a convex problem and can be solved using existing toolboxes, such as CVX. Denote the optimal solution to this problem by $(\tilde{p}_k^{\mathrm{U}},\tilde{p}_k^{\mathrm{R}})$ , and we have proved Proposition~\ref{proposition2}
\end{proof}

Returning to problem \eqref{subP}, after the transmit powers are determined, $n$ and ${\mu}_k$ can be obtained by constraint (\ref{subP}a) and \eqref{OP-min Principle}, respectively, and problem \eqref{subP} has also been solved.

\section{Numerical Simulation}
In our simulations, we assume that there are $K=4$ users and the FAS structure parameters are configured as $N_1=N_2=4$ and $W_1=W_2=1$. Besides, we set the channel model parameters as $\sigma_k^2=\sigma^2=-120\mathrm{dBm}$, $B=5\mathrm{MHz}$, and $B\xi = 0.5\mathrm{Mbps}$. Without loss of generality, we always set $N_1=N_2$, $W_1=W_2$, a equal minimum achievable rate requirement for all users, and equal maximum transmit powers for different users and for different FARs, respectively.  We evaluate and compare the sum rate of our proposed algorithm with several benchmark schemes: \textbf{TAS}, which uses TAS at the FARs to receive and transmit the signals; \textbf{Average bandwidth allocation}, which equally allocates the total bandwidth $B$ to every user; \textbf{Random power control}, which randomly sets the transmit powers for users and the FARs, respectively.

Fig.~\ref{UserNumber} illustrates the sum rate of the system versus the number of users. Without loss of generality, we set the average channel gain of user $k+1$ is greater than that of user $k$. On this basis, sum rates of all schemes increase as the number of users increases, while our proposed algorithm always outperforms the other three schemes. In specific, the `TAS' scheme always performs worst, showing the superiority of the introduction of the FAS. Besides, as more users are included in the system, our proposed algorithm gains more rate improvement than the `Average bandwidth allocation' scheme, demonstrating 
the necessity of adaptive bandwidth allocation.

Fig.~\ref{N} shows the system sum rate versus the number of pre-set ports of the FAS. Our proposed algorithm first holds a same performance with `TAS' scheme when $N_1=N_2=1$, since FAS now degenerates into TAS. As the number of ports increases, our proposed algorithm begins to demonstrate its superiority, always outperforming the other three schemes. Moreover, the growth rate of sum rate in our proposed algorithm gradually decreases with the growth of the number of pre-set ports, and when $N_1$ equals $4$, it can have a relatively large gain compared with `TAS' scheme.

Fig.~\ref{Pr} exhibits the sum rate versus the maximum transmit power of the FAR. The `Proposed' scheme maintains the highest sum rate, exhibiting a steady and robust increase as the transmit power rises. The `Random power control scheme', while initially lower than `TAS' scheme at low power, shows a faster growth rate and surpasses `TAS', eventually outperforming `TAS' at higher transmit powers. The `Average bandwidth allocation' scheme falls between the Proposed and other schemes, showing improved sum rates with increasing transmit power but consistently remaining below the `Proposed' scheme. The performance gap between the `Proposed' scheme and the other benchmark schemes at higher power levels emphasizes the capability of the proposed scheme in converting transmit power to higher sum rates.

\section{Conclusion}
In this paper, we have investigated the FAR-assisted wireless communication system with hybrid relaying scheme selection and formulated a sum-rate maximization optimization problem. To determine the relay scheme selection strategy, we have analyzed the OPs using different relaying schemes and have introduced a Gaussion copula-based method to approximate numerical value of OPs, based on which we determine the OP-minimized principle for FAR nodes to choose relaying schemes. To solve the optimization problem, we have obtained the optimal bandwidth allocation and have solved the transmit power problem with a low complexity method. Simulation results have validated the effectiveness of proposed algorithm.
\bibliographystyle{IEEEtran}
\bibliography{MMM}

\end{document}